\documentclass[12pt]{article}
\usepackage{amsmath}
\usepackage{graphicx,psfrag,epsf}
\usepackage{enumerate}
\usepackage{natbib}
\usepackage{comment}
\usepackage{url} 
\usepackage{hyperref}
\usepackage{graphicx} 
\usepackage{amsthm,amssymb}
\usepackage{bm}
\usepackage{hyperref}
\usepackage{dsfont}
\usepackage[caption=false]{subfig}
\usepackage{tabularx}
\usepackage{float}
\usepackage{algorithm}
\usepackage{algpseudocode}
\usepackage{multirow}
\usepackage{todonotes}

\newtheorem{theorem}{Theorem}[section]

\newtheorem{lemma}[theorem]{Lemma}

\newtheorem{corollary}[theorem]{Corollary}

\newtheorem{definition}[theorem]{Definition}

\usepackage[caption=false]{subfig}

\newcommand{\nn}{ \nonumber \\}

\newlength\aftertitskip     \newlength\beforetitskip
\newlength\interauthorskip  \newlength\aftermaketitskip

\begin{document}


	{
		\title{ A new method for quantifying network cyclic structure to improve community detection
		}
		\author{Behnaz Moradi-Jamei\footnote{\texttt{bm7mp@virginia.edu}. School of Data Science, University of Virginia.  P.O. Box 400249, 
Charlottesville, VA, 22904.}, Heman Shakeri\footnote{\texttt{hs9hd@virginia.edu}. School of Data Science, University of Virginia.  P.O. Box 400249, 
Charlottesville, VA, 22904.}, Pietro Poggi-Corradini\footnote{\texttt{pietro@math.ksu.edu}.  Department of Mathematics, Kansas State University. 138 Cardwell Hall,  Manhattan, KS., 66506.  This material is based upon work supported by the National Science Foundation under Grant No. DMS 1515810.} and\\ Michael J. Higgins\footnote{\texttt{mikehiggins@ksu.edu}. Department of Statistics, Kansas State University. 101 Dickens Hall,  Manhattan, KS., 66506.}$^{~}$\footnote{Corresponding author.}
			}
		\maketitle

	\begin{abstract} 
		A distinguishing property of communities in networks is that cycles are more prevalent within communities than across communities.
		Thus, the detection of these communities may be aided through the incorporation of measures of the local "richness" of the cyclic structure. 
		In this paper, we introduce renewal non-backtracking random walks (RNBRW) as a way of quantifying this structure. 
		RNBRW gives a weight to each edge equal to the probability that a non-backtracking random walk completes a cycle with that edge. 
		Hence, edges with larger weights may be thought of as more important to the formation of cycles.
		Of note, since separate random walks can be performed in parallel, RNBRW weights can be estimated very quickly, even for large graphs.
		We give simulation results showing that pre-weighting edges through RNBRW may substantially improve the performance of common community detection algorithms.
	    Our results suggest that RNBRW is especially efficient for the challenging case of detecting communities in sparse graphs. 

	\end{abstract}
	
	\noindent%

	\maketitle

\section{Introduction}

In many sciences---for example Sociology~\citep{scott2017social}, Biology~\citep{chung2015bridging}, and Computer Science~\citep{hendrickson2000graph}---units under study often belong to communities, and thus, may behave similarly.
Hence, understanding the community structure of units is critical in these sciences, and much work has been devoted to the development of methods to detect these communities~\citep{fortunato2010community}.  

Community detection methods often use the framework of mathematical networks: units are nodes (vertices) in the graph, and edges are drawn between two nodes if the corresponding units interact with each other.
Under this framework, the problem of community detection becomes a graph partitioning problem where units within each block of the partition are ``optimally'' connected under some objective.
Commonly, communities are identified through edge prevalence; edges are more frequently observed between nodes in the same community than between nodes in different communities.

Optimal detection of communities is often $\mathcal{NP}$-hard because the number of possible partitions grow exponentially  as the number of units $n$ grows.
Many (mostly heuristic) methods have been developed for the detection of communities under large $n$ settings.
Two such methods include the Louvain algorithm~\citep{blondel2008fast} and the CNM~\citep*{clauset2004finding} algorithm.
Moreover, the ability of these algorithms to accurately identify communities may be improved through incorporating additional graph structure, which are typically integrated into these algorithms through the use of edge weights~\citep{de2013enhancing,sun2014weighting, khadivi2011network}.


Of note, previous work has shown that measures quantifying the cyclic structure of the network may be useful in detecting communities~\citep{newman2003structure, radicchi2004defining, kim2005cyclic, vragovic2006network, zhang2008clustering, shakeri2017network}.
The intuition is that, if edges are more prevalent within a community than across communities, cycles (especially small cycles) may be more prevalent within communities as well. 
Hence, edge weights that incorporate information about the density of cycles within a graph may be useful in detecting communities within that graph.
However, some of these methods may be computationally prohibitive to implement as network sizes grow very large.


On the other hand, random walks have been proposed as a computationally efficient tool for uncovering the structure of networks. 
When the density of edges is higher within communities than across communities, random walks will spend a majority of their time traveling within communities~\citep{hughes1995random, lai2010enhanced}. 
Hence, network structure may be uncovered by following the paths of these random walks.
An important variant is the non-backtracking random walk (NBRW), in which the random walk is prohibited from returning back to a node in exactly two steps~\citep{alon2007non, fitzner2013non}.
While work on NBRW currently focuses on its fast convergence rate, we exploit another useful property of NBRW---its ability to identify cyclic structure.  

In this paper, we consider the process in which a NBRW is performed until it forms a cycle, at which point the walk terminates and restarts---we call this process the renewal non-backtracking random walk (RNBRW).
We are particularly interested in the \textit{retracing probability} of an edge---the probability that the edge completes the cycle in one iteration of RNBRW. 
Intuitively, edges with higher retracing probabilities are more critical to the formation of a cycle.
Hence, the cyclic structure of a network may be incorporated into community detection algorithms through weighting edges by their retracing probability.
Although analytically obtaining exact values for retracing probabilities is difficult, repeated iterations of RNBRW can yield accurate estimates of these probabilities.  
Additionally, these iterations can be performed in parallel, allowing these probabilities to be estimated extremely fast, even for networks containing millions of nodes.  

We show through simulation that weighting edges by their estimated retracing probabilities through RNBRW improves substantially the ability for Louvain and CNM methods to detect communities.
The improvement is especially noticeable in the case of sparse graphs.
Additionally, we show that using Louvain and CNM with these weights may offer comparable performance to that of state-of-the-art community detection methods with considerably less computation.


\section{Notation and preliminaries}\label{sec:prelim}
Let $G=(V,E)$ denote a graph with node set $V$ and edge set $E$. 
We do not currently make any restrictions on whether $G$ is directed or undirected. 
Let $n = |V|$ and $m = |E|$. 
If $G$ is undirected, we will denote the edge between vertices $i$ and $j$ by $ij$, and if $G$ is directed, we denote the edge traveling from $i$ to $j$ by $\vec{ij}$.
For convenience, we may refer to edges in an undirected graph $G$ using the latter notation; 
in this case, $\vec{ij}$ and $\vec{ji}$ refer to the same edge $ij$.
Edges $\vec{ij} \in E$ may be weighted; let $w_{\vec{ij}}$ denote the value of this weight.
For ease of notation, we assume an unweighted graph has default edge weights $w_{\vec{ij}} = 1$.
An \textit{adjacency matrix} $A$ for a graph $G$ is an $n \times n$ matrix 
where element $A_{ij} =1$ if edge $\vec{ij} \in E$ and $A_{ij} = 0$ otherwise. 

The (possibly weighted) \textit{degree} of node $i$, denoted $d_i$, is the sum of the edge weights that begin from node $i$. 
\begin{equation}
    \label{eq:definedi}
   d_i \equiv \sum_{j = 1}^n w_{\vec{ij}}
\end{equation}
We assume edge weights are scaled so that
\begin{equation}
    \sum_{i=1}^n d_i = 2m.
\end{equation}
Note that, in the unweighted case, $d_i$ is simply the number of edges beginning from node $i$. 

We deviate from conventional literature by defining a walk by the edges it traverses; a \textit{walk} $(\overrightarrow{v_0v_1}, \overrightarrow{v_1v_2}, \ldots, \overrightarrow{v_{k-1}v_k})$ of length $k$ is a vector of edges  $\overrightarrow{v_{\ell-1}v_{\ell}} \in E$, $\ell = 1,2,\ldots, k$ connecting (possibly non-distinct) nodes $v_0,v_1, v_2, \ldots,v_k \in V$.
A \textit{random walk} $(e_1, e_2, \ldots, e_k)$ is a walk such that:
\begin{align}
    &P\left(e_{\ell+1} = \overrightarrow{v_{\ell}v_{\ell+1}}  | e_{\ell} = \overrightarrow{v_{\ell-1}v_{\ell}}, \ldots, e_1 = \overrightarrow{v_{0}v_{1}}\right)\nonumber \\
     =\;&  
     P\left(e_{\ell+1} = \overrightarrow{v_{\ell}v_{\ell+1}} | e_{\ell} = \overrightarrow{v_{\ell-1}v_{\ell}}\right) = \frac{w_{\overrightarrow{v_{\ell}v_{\ell+1}}}}{d_\ell}
\end{align} 
A walk $c$ is a \textit{simple cycle} 
if $v_{0} = v_{k}$ and $v_0,v_1, v_2, \ldots,v_{k-1}$ are distinct.

The nodes $V$ are partitioned into $q$ communities, numbered 1 through $q$.
Let $g_v \in \{1, \ldots, q\}$ denote the community to which $v$ belongs and $\mathbf g = (g_1, g_2, \ldots, g_n)$.
In problems of community detection, we wish to uncover the true label $g_v$ for each node $v \in V$.

\subsection{Community detection through maximizing modularity}
One popular approach for the problem of community detection is to choose communities that corrrespond to the optimal solution of an integer programming problem.  
A common objective function for community detection is 
the graph modularity~\citep{newman2003structure}: 
\begin{equation}\label{eq_Modularity}
    M(\mathbf g; A) = \frac{1}{2m}\sum_{i,j}\left(A_{ij}-\frac{d_i d_j}{2m} \right) \delta_{g_i g_j}.
\end{equation}
Here, $\delta_{g_i g_j}$ is the Dirac delta function equal to 1 if
$g_i = g_j$---that is, if $i$ and $j$ belong to the same community.
Intuitively, for a given community label $\mathbf g$, the modularity function measures how much within-group and across-group edge formation deviates from independence.
Better choices of communities $\mathbf g$ correspond to larger values of $M(\mathbf g)$.  
Choosing communities $\mathbf g$ that maximize the modularity---or almost any other graph partitioning objective---is an $\mathcal{NP}$-hard problem; hence, this approach often focuses on finding heuristic or approximately optimal solutions~\citep{brandes2008modularity}.
We now detail two such approaches, the CNM method~\citep{clauset2004finding} and the Louvain method~\citep{blondel2008fast}.

\subsubsection{CNM method}
The CNM method is a greedy algorithm in which communities are iteratively merged so that each merger maximizes the change in modularity.
Let $\Delta Q_{g_ig_j}$ denote the change in modularity obtained when communities $g_i$ and $g_j$ are merged together.
For CNM, each node initially begins in its own singleton community.
In this case, the change of modularity obtained from merging nodes $i$ and $j$ into a single community is
\begin{equation}\label{eq_CNMstep}
    \Delta Q_{ij} = \begin{cases} 
    1/{m}-2d_id_j/(2m)^2 & \text{if }A_{ij}\neq 0 \\
    0 & \text{otherwise}
\end{cases}
\end{equation}
Once the initializations are complete, the algorithm repeatedly selects the merger that maximizes modularity, and then updates the $\Delta Q$ values, until no further mergers are possible.
CNM then selects communities according to the largest found modularity.
The computational complexity of the CNM algorithm is
$O(n^2\log n)$.
\subsubsection{Louvain method}
The Louvain method is divided into two phases.  
The first phase initializes each node into its own community 
and considers increasing modularity locally by changing a node's community label to that of a neighboring community. 
For example, the modularity change obtained by moving singular node $i$ from its community to a neighboring community $g$ is calculated by:
\begin{equation}\label{eq_LouvStep}
\begin{split}
\Delta Q_{-i,g}={\bigg [}{\frac {\Sigma _{in}+d_{i,in}}{2m}}-{\bigg (}{\frac {\Sigma _{tot}+d_{i}}{2m}}{\bigg )}^{2}{\bigg ]}
-{\bigg [}{\frac {\Sigma _{in}}{2m}}-{\bigg (}{\frac {\Sigma _{tot}}{2m}}{\bigg )}^{2}-{\bigg (}{\frac {d_{i}}{2m}}{\bigg )}^{2}{\bigg ]}
\end{split}
\end{equation}
where $\Sigma_{in}$  is sum of internal edge weights within $g$, $ \Sigma _{tot}$ is the sum of all edge weights incident to nodes in $g$, $d_{i,in} $ is the sum of the weights of edges between node $i$ and nodes in $g$, 
and $m$ is the sum of the edge weights of the entire network. 
For each node $i$, this change in modularity is computed for each community neighboring $i$, and $i$ is then assigned the community label that corresponds with the largest change.  
If no positive change in modularity is possible, then $i$ does not change its community label.  

In the next phase, the algorithm uses the final network from the first phase to generate a new network.  
Each community from phase one is condensed into a single node, and edges between two (possibly non-distinct) communities are condensed into a single edge with weight equal to the sum of the corresponding between-community edge weights.
In particular, self-loops in the new network may be formed if the corresponding community in the previous network has more than one node.
Phase one is then applied to this new network.

These two phases are repeated iteratively until no local improvement in the modularity is possible.
The Louvain algorithm examines fewer mergers (only adjacent ones) compared to CNM, and hence, the time complexity is reduced---it appears to run in linearithmic time for sparse graphs \citep{lancichinetti2009community}.

%
%
%
%
%

\subsection{Using edge weights to improve community detection}
 The quality of community detection may be improved by adding a pre-processing step to weight edges of the graph.  
 We briefly go over some proposed weighting schemes in the literature.

\subsubsection{Weights based on random walks}
Inspired by message propagation, \citet{de2013enhancing} introduces a method called weighted edge random walk-K path (WERW-Kpath) that runs a $k$-hop random walk and weights edges on the walk based on their ability in passing a message. 
WERW-Kpath initially assigns to each edge $e\in E$ a weight $w_e = 1$.
At each iteration, WERW-Kpath chooses a node at random and and runs a biased random walk of at most $k$ steps; the probability of traveling from node $i$ to $j$ is fraction of times that $j$ has been visited from node $i$ over the total number of times $i$ has been visited. 
Final edge weights are obtained after many iterations of WERW-Kpath.
\citet{lai2010enhanced} suggests a similar approach for using random walks to explore local structures of communities, and ultimately obtain an edge weight for edge $ij$ as the cosine distance between $i$ and $j$.


\subsubsection{Preprocessing the graph by analyzing cyclic topologies}\label{sec:LitRevWeighting}
The idea of using cyclic structure to identify communities relies on the fact that, if edges are more prevalent within communities than across communities, then flows will also tend to stay within communities. 
For example,  \citet{klymko2014using} focused on using triangles (cycles of length 3) to improve community detection in directed networks by weighting  the edges based on \textit{3-cycle cut ratio}.
\citet{radicchi2004defining} incorporate the importance of triangles using the edge clustering coefficient.
This method is similar to \citet{newman2004finding} in which at each iteration the edge with the smallest clustering coefficient is removed.
The complexity of the algorithm in~\citet{radicchi2004defining} is $O(m^4/n^2)$---roughly $O(n^2)$ for sparse graphs. 
This algorithm may perform poorly in graphs with few short cycles.

\citet{castellano2004self} modified the edge clustering coefficient for the weighted networks, in which number of cycles is multiplied by the edge weight. 
\citet{zhang2008clustering} extended this to bipartite graphs with 
cycles of even length.
\citet{vragovic2006network} introduced the node loop centrality measure such that communities are built around nodes with high centrality in graph cycles---this algorithm has time complexity of $O(nm)$.
Communities with rich cyclic structures--many short loops--are known to have higher quality; flows tend to stay longer and information disseminates faster in their nodes.

\citet{shakeri2017network} suggest weighting edges using a method called \textit{loop modulus} (LM). 
LM finds a probability distribution over all cycles such that two random cycles from this distribution have "minimum edge overlap."
The weight of an edge is the probability that this edge is overlapped.

\subsubsection{Other weighting method based on local and global network structure}\label{sec:k_path}
The use of global measures for determining edge weights requires the knowledge of entire networks. 
While these edge weights may be helpful in community detection, their use requires substantial computation, and sometimes the weights lack the necessary specificity to the community structures.
Local measures are often much faster to compute and may estimate local network subtlety that global methods cannot, but are unable to obtain thorough knowledge of global network structure.  
We now summarize the use of global and local measures in the literature.

\citet{newman2004finding} propose to use of edge-betweeness centrality (EBC) as a way to evaluate the community structure of networks. 
EBC was introduced by \citet{freeman1977set} and, for each edge $e$, measures the fraction of geodesic paths between all pairs of nodes that passes through $e$. 
\citet{khadivi2011network} proposed to use common neighbor ratio in addition to EBC to balance the local and global measurements and weigh the graph edges. 
This will allow for smaller weights of edges formed between clusters and also facilitates greedy algorithms in identifying the hierarchical structure of the clusters. 
The runtime complexity of this weighting strategy is dominated by computing the edge-betweenness centrality which is in $\mathcal{O}(nm)$ \citep{brandes2001faster} 
This runtime is prohibitive for large networks.
Furthermore, tuning the hyperparameters requires the design of efficient heuristics to improve the performance of the weighting algorithm.
\citet{zhang2015novel} focused on the problem of local and global weighting of the edges to improve the separation of communities. They proposed to measure the similarity of each pair of nodes based on their similarity to other nodes through an iterative function called SimRank.

\section{Renewal non-backtracking random walk (RNBRW)} \label{sec_3}



We now introduce our method for weighting edges called the renewal non-backtracking random walk (RNBRW).
This method is effective in capturing network cyclic structure while maintaining the computational advantages given by random walk methods. 
We begin with some definitions.

Recall the definition of a random walk in Section~\ref{sec:prelim}.
A non-backtracking random walk (NBRW) is a random walk that cannot return to a node $i$ in exactly two steps. 
That is, for all nodes $j$,
\begin{eqnarray}
    P(e_{\ell+1} = \vec{ji} | e_{\ell} = \vec{ij}) = 0.
\end{eqnarray}
Note that, while NBRW is not Markovian on the nodes of the graph, we can make NBRW Markovian on the edges by replacing each undirected edge $uv$ in $E$ with two directed edges $\vec{uv}$ and $\vec{vu}$. 

For an unbiased NBRW, the transition probability is \citep{kempton2016non}
\begin{equation}\label{eq_NBRW_transition}
\Pr(\vec{jk},\vec{ij})=
 \begin{cases} 
1/(d_j-1),& \text{ if } \vec{ij}, \vec{jk} \in E \text{ and } k \neq i, \\
0, & \text{otherwise},
\end{cases}
\end{equation}
where $d_j$ is the out-degree of node $j$.
A primary benefit of NBRW is that it explores the graph more efficiently than random walks; using NBRW, the proportion of time spent on each edge converges faster to the stationary distribution.

In previous works, such as \citet{de2013enhancing}, the random walk has been used to quantify the idea of information passing and to identify nodes or edges that are more capable of passing information.
Both random walk and NBRW behave similarly in the long run and their stationary distribution depends only on the degrees of nodes \citep{kempton2015high}. 
To reduce this dependability on node-degree, we keep the walk length small by stopping the NBRW.
However, instead of stopping NBRW after a given number of iterations, we instead make a stopping rule that may help in uncovering cyclic structure in the network, thereby obtaining the renewal non-backtracking random walk. 

\begin{definition}[Renewal non-backtracking random walk (RNBRW)]
A renewal non-backtracking random walk (RNBRW) is a NBRW that terminates and restarts once the walk completes a cycle.
Each NBRW begins at a directed edge selected completely at random.
Each new run of a NBRW is an iteration of RNBRW.
\end{definition}
\begin{definition}[Retraced edge]
	The retraced edge for a run of RNBRW is the edge that forms a cycle. 
	More precisely, supposing that $j$ is the first node revisited by a NBRW, and supposing that the NBRW travels to $j$ from $i$, then $\vec{ij}$ is the retraced edge.  
\end{definition}
\begin{definition}[Retracing probability]
	The retracing probability $\pi_e$ of an edge $e$ is the probability that $e$ is retraced in one iteration of RNBRW given that the NBRW terminates by forming a cycle. 
	The retracing probability $\pi_{ij}$ of an undirected edge $ij$ is obtained by summing the retracing probabilities for both $\vec{ij}$ and $\vec{ji}$.
\end{definition}

Figure~\ref{fig:RetracedNBRW} gives an example of one RNBRW iteration on a graph.
Heuristically, edges with large retracing probabilities are important to the formation of cycles.
As discussed in Section~\ref{sec:LitRevWeighting}, the presence of cycles is a helpful indicator in discovering communities.
Thus, it seems reasonable that quantifying the cyclic structure by the retracing probabilities may lead to improved performance of algorithms designed for detecting communities.
Figure~\ref{fig:HouseNBRW} gives undirected retracing probabilities for a small graph.

\begin{figure}
	\centering
	\includegraphics[clip, angle = 90, width=.4\columnwidth]{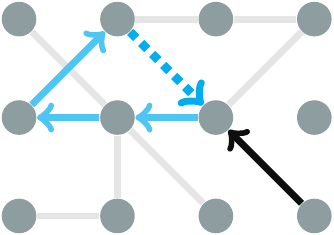}%
	\caption{One iteration of RNBRW on a graph. Gray lines denote edges not traversed by the NBRW.  The walk begins with the edge in black.  The walk completes after forming the cycle in blue. The dotted edge represents the retraced edge.}
	\label{fig:RetracedNBRW}
\end{figure}

	\begin{figure}
		\centering
		\includegraphics[clip,width=.5\columnwidth]{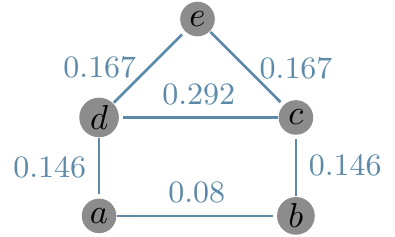}%
		\caption{Undirected retracing probabilities for a house graph with 5 nodes and 6 edges. 
		Edges with more overlapping loops have higher retracing probabilities.  
		For this graph, edge $dc$ appears to be the most important edge for the formation of cycles. 
		}\label{fig:HouseNBRW}
	\end{figure}

The retracing probabilities $\pi_e$ for each edge $e$ may be difficult to compute analytically, even for small graphs.  
Instead, as we show in Section~\ref{subsec:asymptoticprops} these probabilities can be estimated precisely by running many iterations of RNBRW.
Precisely, the estimated retracing probability for edge $e$, denoted $\hat \pi_e$, is the fraction of RNBRW runs for which $e$ is the retracing edge.
Note that $\sum_{e} \pi_e = \sum_{e} \hat \pi_e = 1$.

A particularly useful property of RNBRW is that iterations are mutually independent.
In particular, unlike other random walk methods (\textit{e.g.} \citet{de2013enhancing}), each run is not affected by weights obtained from previous runs.
Hence, multiple RNBRW instances can be run in parallel, and thus, retracing probabilities can be estimated quickly, even for large graphs.
For example, in \ref{sec:results}, we apply RNBRW on graphs containing one-million nodes.

Finally, note that a RNBRW iteration may terminate in one of two ways: either the run forms a cycle or the NBRW reaches a node with degree 1.
In the latter case, the iteration is discarded---no retracing edge is recorded---and a new iteration of RNBRW is started.
Fortunately, the probability that this case occurs may become vanishingly small as the number of nodes $n$ gets large.
For example, for Erd\H{o}s-R\'enyi graphs, the probability that a run terminates by visiting a node with degree 1 decreases exponentially with $n$~\citep{tishby2017distribution}.

\subsection{RNBRW for community detection}

RNBRW can be implemented to improve community detection algorithms as follows.
First, the estimated undirected retracing probability $\hat \pi_e$  is obtained for each edge $e$ through performing many iterations of RNBRW.
Then, each edge $e$ in the graph given the weight $w_e = 2m\pi_e$, where $m$ is the number of edges.
Next, the weighted degree $d_i$ in~\eqref{eq:definedi} is computed for each node $i$.
Note that multiplying $\hat \pi_e$ by $2m$ ensures that the sum of degrees $\sum_{i=1}^n d_i$ is the same when moving from unweighted edges to weighted edges. 
These weighted degrees are then plugged into the modularity function~\eqref{eq_Modularity}.
Finally, community detection algorithms such as CNM or Louvain are then performed using the weighted modularity function.

Figure~\ref{fig:weightLFR} demonstrates the usefulness of RNBRW weighting in community detection.
This figure shows an LFR benchmark graph~\citep{lancichinetti2008benchmark} with four communities before and after weighting RNBRW.
Weights of edges formed within communities are much larger on average than weights of edges across communities.

\begin{figure}
	\centering
	\subfloat[]{		
		\includegraphics[clip,width=.35\columnwidth]{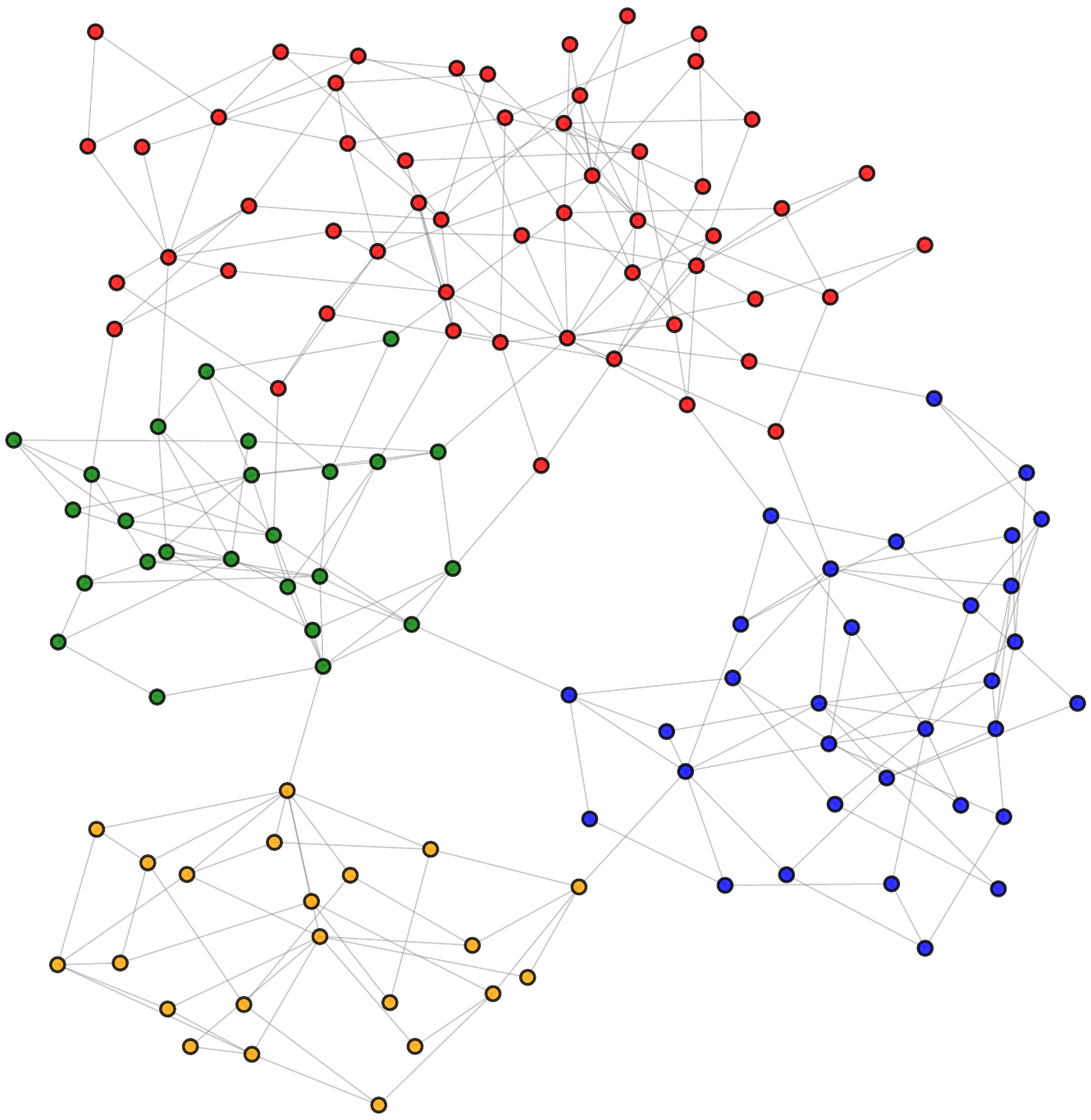}%
	}~~~~~~~~~
	\subfloat[]{%
		\includegraphics[clip,width=.35\columnwidth]{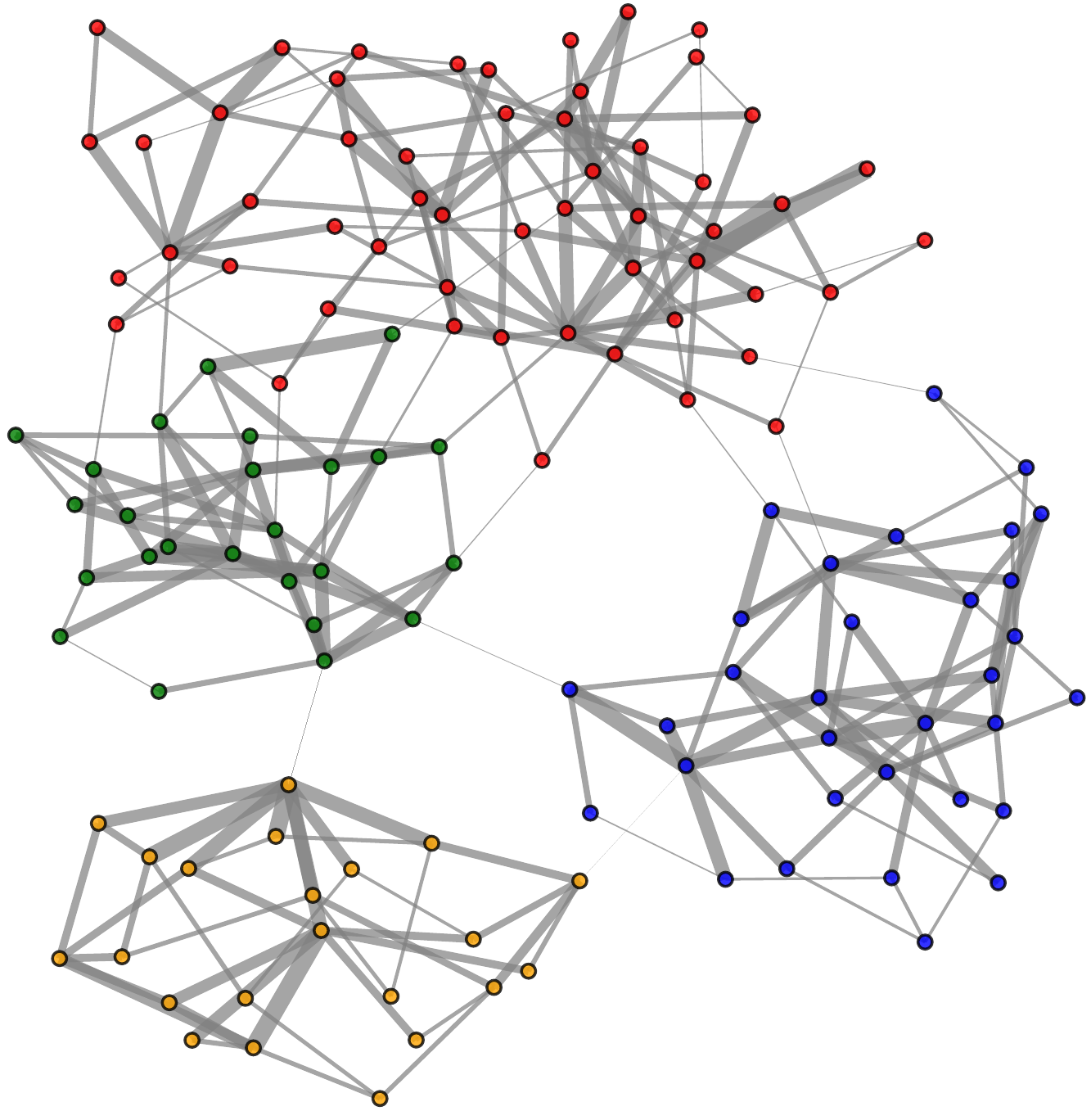}%
	}\\


	\caption{(a) An LFR benchmark graph. (b) The same graph weighted by RNBRW. Observe that within-community edges have larger weights than across-community edges.} 
	\label{fig:weightLFR}
\end{figure}

Figure~\ref{fig:iteration} gives a simple illustration of how many iterations of RNBRW is necessary in order to obtain satisfactory detection of communities after weighting.
As the figure indicates, 
there is a large improvement in the detection of communities when the number of walkers is on the order of the number of edges in the graph. 
Running additional walkers will still improve this precision;
however, the corresponding improvement in community detection is fairly small.
A theoretical justification for this number of iterations
can be found in Section~\ref{subsec:asymptoticprops}.


	\begin{figure}
	\centering
	\includegraphics[clip,width=.5\columnwidth]{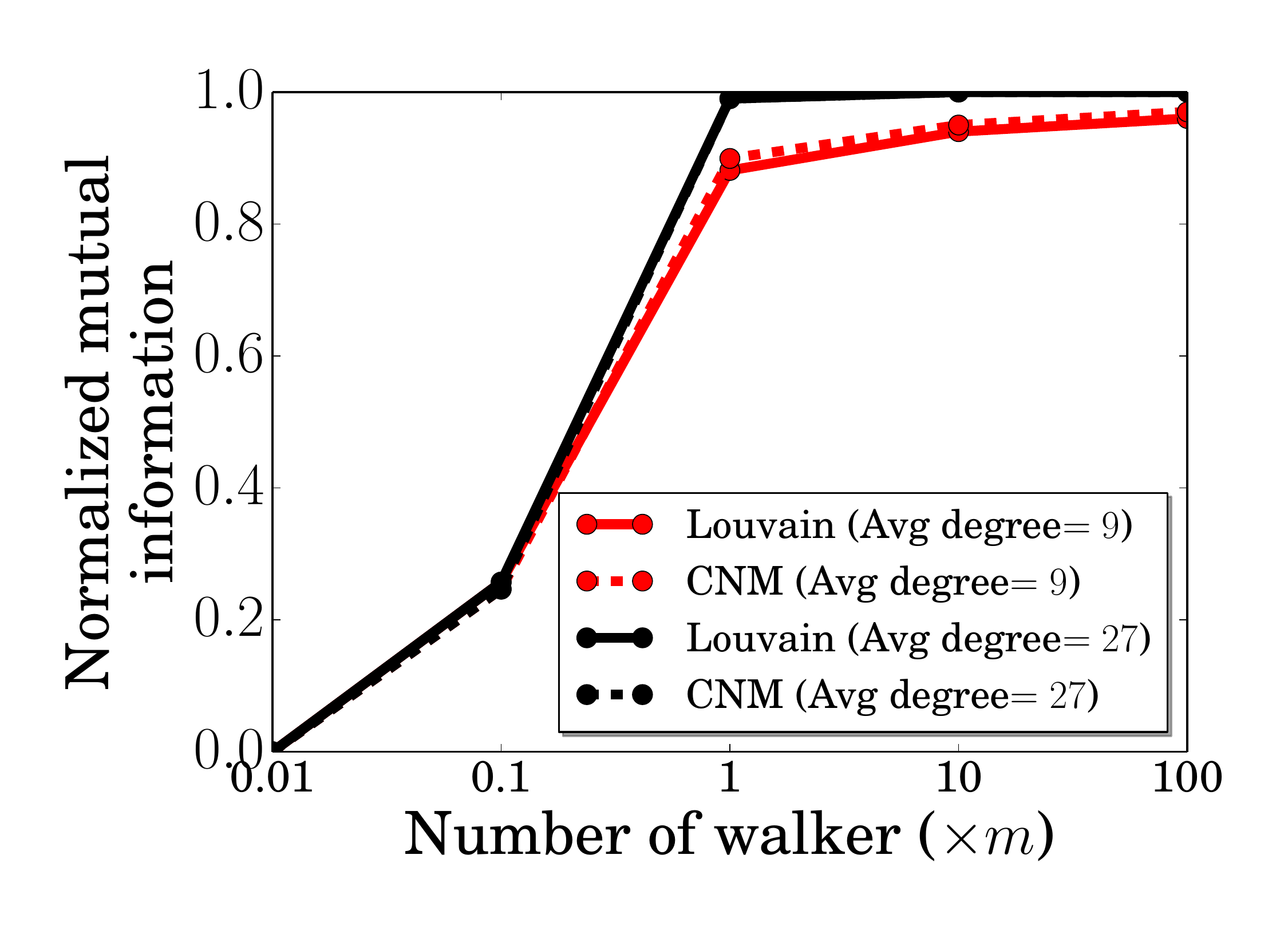}
	\caption{Performance of community detection with RNBRW weighting increasing the number of walkers in both CNM and Louvain using RNBRW weighting for LFR networks with $n = 10,000$ and average degrees $9$ and $27$.}
	\label{fig:iteration}
\end{figure}

\

\subsection{Algorithm for RNBRW}

For completeness, we now describe in detail the RNBRW algorithm. 
Each run of the algorithm runs a NBRW and returns the retracing edge (if any) as the sample of retracing edge. 
The number of times each edge has been retraced by a NBRW is approximately proportional to its retracing probability.
The algorithm only requires knowledge the graph and each run is independent of the other. 
Therefore, we can collect samples in parallel leading to fast convergence. 
Each run of the algorithm is comprised of the following steps:
 \begin{itemize}
	\item[1] Choose a random edge $\overrightarrow{v_0v_1}$ in $E$.
	\item[2] Form the walk $w = \left(\overrightarrow {v_0v_1}\right)$. 
	\item[3] (For  $k=\{1,\cdots\}$) The walker continues her walk from ${v_k}$ to a neighboring node $v_{k+1}\neq v_{k-1}$.
	\item[4] If $v_{k+1}$ has degree 1, return immediately to Step~1.  
	If $v_{k+1}$ is already in $w$, return $\overrightarrow{v_kv_{k+1}}$ as the retracing edge and return to Step~1.   
	Otherwise add $v_{k+1}$ to $w$ and go to Step 3 incrementing $k = k+1$.
\end{itemize}
In Algorithm~1, we present the pseudo-code for RNBRW.
By employing a swarm of walkers and recording their retracing edges, we are able to get an accurate estimate of the retracing probability for each edge.
Note that these walkers can be initialized independently from each other and retraced edges can be collated at the very end of the process. 
Therefore, one can execute this process as an array of jobs on a cluster of computers efficiently. 

\begin{figure}
	\begin{minipage}{\linewidth}\label{algNBRW}
		\begin{algorithm}[H]
			\caption{Algorithm for  RNBRW from a random directed edge $\vec{uv}$.}
			\begin{algorithmic}[1]
				\State \textit{walk} $\leftarrow$empty set
				\State \textbf{choose} edge $\vec{uv}$ at random
				\State \textbf{add} $\vec{uv}$ to \textit{walk}
				\While{True}
				\State \textit{nexts}$\leftarrow$neighbors of $v$
				\State \textbf{remove} $u$ from \textit{nexts}
				\State \textbf{if} nexts is empty \textbf{break}
				\State \textit{next} $\leftarrow$\textbf{choose} a node from \textit{nexts} randomly
				\If {\textit{next }$\in$\textit{walk}}
				\State \textbf{return} ($v$,\textit{next}) as the retracing edge
				\State \textbf{add} \textit{next} to \textit{walk}
				\State $u\leftarrow v$
				\State $v \leftarrow $\textit{next}
				\EndIf
				\EndWhile
			\end{algorithmic}
		\end{algorithm}
	\end{minipage}
\end{figure}


\section{Asymptotic properties of RNBRW \label{subsec:asymptoticprops}}

We now give proof that the RNBRW estimates for the retracing probabilities converge almost surely to the true retracing probabilities.
The proof also gives insight as far as how many iterations are necessary in order to obtain a good estimate of these probabilities.
We begin with a few definitions.

Suppose there are $\rho$ iterations of RNBRW that conclude by returning a retracing edge.
For a given directed edge $\vec e$, define $Y_{\vec e,k}$ as a random variable indicating that the $k$th iteration returns edge ${\vec e}$ as the retracing edge.
That is,
\begin{equation}
	Y_{\vec e,k}=
		\begin{cases}
		1,&  \text{if the $k$th RNBRW returns $e$ as the retracing edge}, \\
		0, &  otherwise.
		\end{cases}
\end{equation}
Note that $Y_{\vec e,k}$ is a \textit{Bernoulli}$(\pi_{\vec e})$ random variable and that
\begin{equation}
    \overline Y_{\vec e} = \frac{1}{\rho}\sum_{k = 1}^\rho Y_{\vec e,k} = \hat \pi_{\vec e}.
\end{equation}

We first make a statement about how close any arbitrary estimated retracing probability comes to the true probability.
\begin{lemma}\label{lemma:justhoeff}
    \begin{equation}
    \label{eq:justhoeff}
        P(|\overline{Y_{\vec e}}-\pi_{\vec e}| \geq \epsilon) \leq 2 \exp(-2\rho \epsilon^2).
    \end{equation}
\end{lemma}
\begin{proof}
    This result follows from direct application of Hoeffding's inequality~\citep{hoeffding1963probability}.
\end{proof}

We now obtain a result on how estimates of undirected retracing probabilities uniformly converge to the true probabilities.
Suppose that, in the original NBRW setup, undirected edge $e$ is replaced with directed edges $\vec e_1$ and $\vec e_2$. 
Hence, $\pi_{e} = \pi_{\vec e_1} + \pi_{\vec e_2}$, and
this undirected probability is estimated by $\overline Y_e = \overline Y_{\vec e_1} + \overline Y_{\vec e_2} = \hat \pi_{\vec e_1} + \hat \pi_{\vec e_2}$.
\begin{lemma}\label{lemma:maxsumhoeff}
    \begin{equation}
    \label{eq:maxsumhoeff}
        P(\max_{e} |\overline{Y_{e}}-\pi_{e}| \geq \epsilon) \leq 4 m \exp(-\rho \epsilon^2/2).
    \end{equation}
\end{lemma}
\begin{proof}
To see this, note that
\begin{equation}
    \label{eq:halfthehoeff}
	P(|\overline{Y_{e}}-\pi_{e} |\geq \epsilon)=
		P(|\overline{Y}_{\vec e_1}-\pi_{\vec e_1} + \overline{Y}_{\vec e_2}-\pi_{\vec e_2} |\geq \epsilon) \leq
		P\left(\bigcup_{\ell = 1}^2\left\{|\overline{Y}_{\vec e_\ell}-\pi_{\vec e_\ell} |\geq \epsilon/2\right\}\right).
\end{equation}
Thus, by~\eqref{eq:halfthehoeff} and Lemma~\ref{lemma:justhoeff}, we have
    \begin{eqnarray}
		P(\max_{e}|\overline{Y_{e}}-\pi_{e} |\geq \epsilon) &=& 
			P\left(\bigcup_{e}\left\{\left|\overline{Y_{ e}}-\pi_{e} \right|\geq \epsilon\right\}\right) \nn &\leq&
			P\left(\bigcup_{e}\bigcup_{\ell = 1}^2\left\{\left|\overline{Y}_{ \vec e_{\ell}}-\pi_{\vec e_\ell} \right|\geq \epsilon/2\right\}\right)	\nn
		&\leq& \sum_{{e}}\sum_{{\ell = 1}}^2 P\left(\left|\overline{Y}_{\vec e_\ell}-\pi_{\vec e_\ell} \right|\geq \epsilon/2\right)\nn &\leq&
		 4m \exp(-\rho \epsilon^2/2).
    \end{eqnarray}
\end{proof}
\noindent
That is, we have shown that the estimated retracing probabilities converge exponentially fast to the true retracing probabilities as number of iterations of RNBRW (that terminate by forming a cycle) increases.  

Lemma~\ref{lemma:maxsumhoeff} also gives insight as far how many iterations of RNBRW are necessary to obtain good estimates of the retracing probability.  
We now give a result that demonstrates how accurate these estimates are if the number of RNBRW runs is approximately equal to the number of edges in the graph.
\begin{corollary}
    Given $\rho = m$ iterations of RNBRW, the estimated retracing probabilities are guaranteed to be within $\sqrt{\frac{4\log(2m)}{m}}$ of the true probabilities
    with probability no smaller than $(m-1)/m$.
    That is,
    \begin{equation}
        	P\left(\max_{e}|\overline{Y_{e}}-\pi_{e} |\geq \sqrt{\frac{4\log(2m)}{m}}\right) \leq \frac{1}{m}.
    \end{equation}
\end{corollary}
\begin{proof}
    To see this result, set $\rho = m$ and $\epsilon = \sqrt{\frac{4\log(2m)}{m}}$ in~\eqref{eq:maxsumhoeff}.
\end{proof}

Finally, we show that the estimated retracing probabilities converge almost surely and uniformly to the true retracing probabilities.
\begin{theorem}[Strong and uniform consistency]
    As the number of iterations of RNBRW that terminate in a cycle $\rho \to \infty$, then $\hat \pi_e \overset{a.s.}{\to} \pi_e$ for each undirected edge $e$.  Moreover this convergence is uniform across all edges $e$.
\end{theorem}
	\begin{proof}
	Using~\eqref{eq:maxsumhoeff}, we observe that
	\begin{eqnarray}
		\lim_{\rho \to \infty}\sum_{n = 1}^\rho P(\max_{e} |\overline{Y_{e}}-\pi_{e}| \geq \epsilon) \leq
			\lim_{\rho \to \infty}\sum_{n = 1}^\rho 4 m \exp(-\rho \epsilon^2/2) < \infty.
	\end{eqnarray}	 
	The theorem then follows directly from the Borel-Cantelli Lemma \citep{rosenthal2006first}.
	\end{proof}

\section{Simulation results}\label{sec:results}
We investigate the performance of two community detection methods, Louvain and CNM, with and without the proposed weighting methods.
In addition to RNBRW, we consider these preprocessing weighting methods: WERW-Kpath, SimRank, Loop modulus, and the weighting method considered in~\citet{khadivi2011network}.
Additionially, we compare the performance from preprocessing with RNBRW to other scalable community detection algorithms such as
Info map \citep{rosvall2007maps}, Label propagation \citep{raghavan2007near}, Edge betweenness \citep{girvan2002community}, Spin glass \citep{reichardt2006statistical}, and Walk trap \citep{pons2005computing}. 
To perform our comparisons, we apply the community detection algorithms on LFR benchmarks \citep{lancichinetti2008benchmark} and measure their similarity to the ground truth data using normalized mutual information (NMI)  \citep{danon2005comparing}. 

\subsection{Comparison of weighting methods}


Figure~\ref{fig:InfoMixModulusNBRW} compares the performance in detecting communities for various weighting methods coupled with the Louvain algorithm as the mixing parameter $\mu$ increases.
Figure~\ref{fig:Khadivi_L} compares these methods for sparse graphs (average degree is $\approx \log(n)$) given a fixed mixing parameter but as the number of nodes increases.
In all cases, preprocessing the graph with RNBRW substantially improves the detection of communities when compared to the unweighted algorithm.
Additionally, the improvement provided by RNBRW is substantially larger than that from the WERW-KPath, SimRank, and the Khadavi \textit{et.~al.} methods.
The performance of RNBRW seems comparable to that of loop modulus.  
However, a good comparison between RNBRW and loop modulus is difficult as current methods for computing loop modulus require prohibitive computational cost for graphs containing more than a couple thousand nodes. 
Hence, loop modulus is omitted from 
Figure~\ref{fig:Khadivi_L}.
Appendix \ref{app_CNM} gives these results for the CNM algorithm.

\begin{figure}
	\centering
	\includegraphics[width=.60\textwidth]{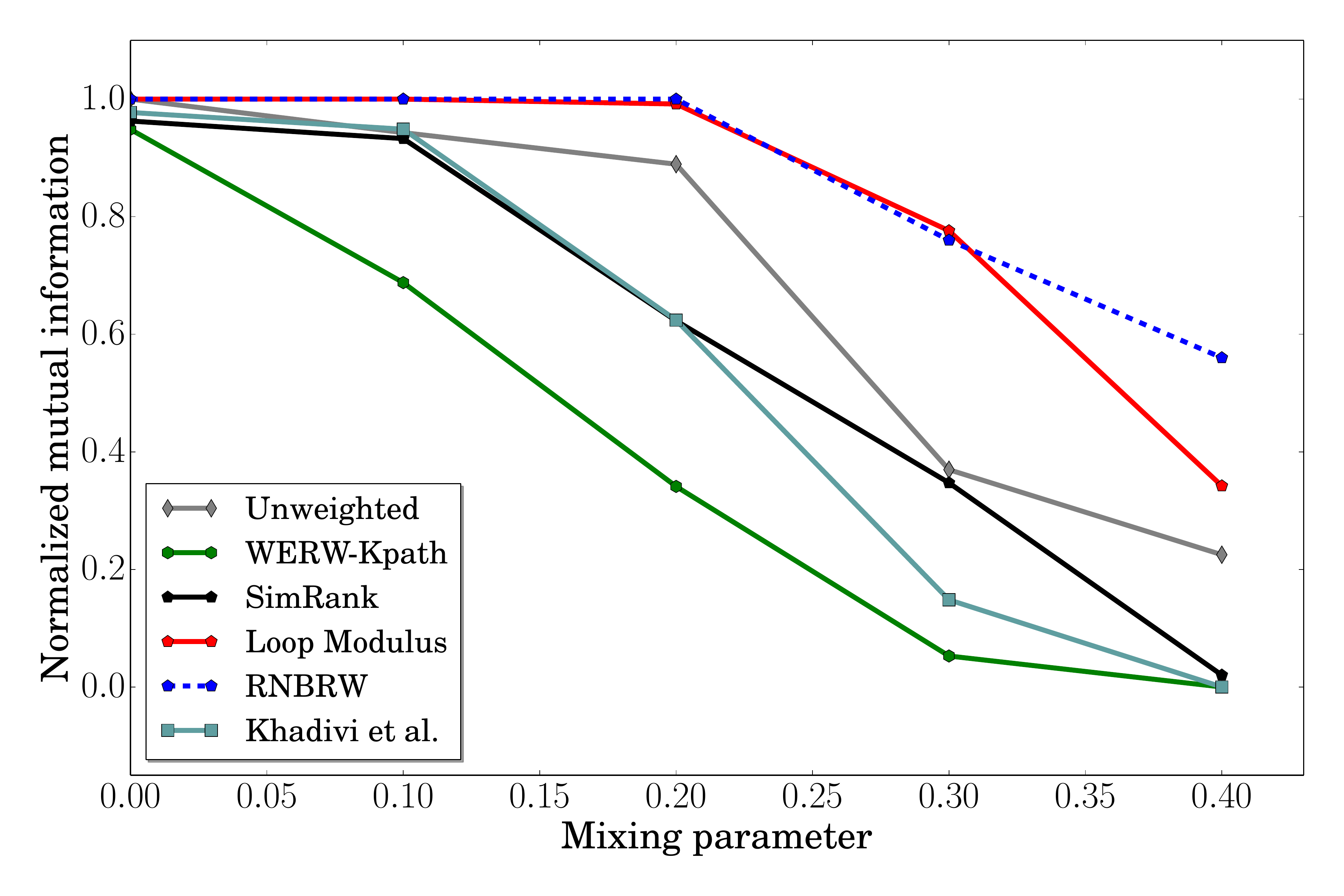}
	\caption{
	Performance analysis of community detection when coupling weighting methods with the Louvain algorithm. 
	We use LFR benchmark networks with  $n = 500$, average degree $7$, and community sizes ranging from $30$ to $70$.
		The mixing rate $\mu$ adjusts the ratio of within-communities links over all links.}\label{fig:InfoMixModulusNBRW}
			
\end{figure}

	\begin{figure}
		\centering

			\includegraphics[clip, width=.60\columnwidth]{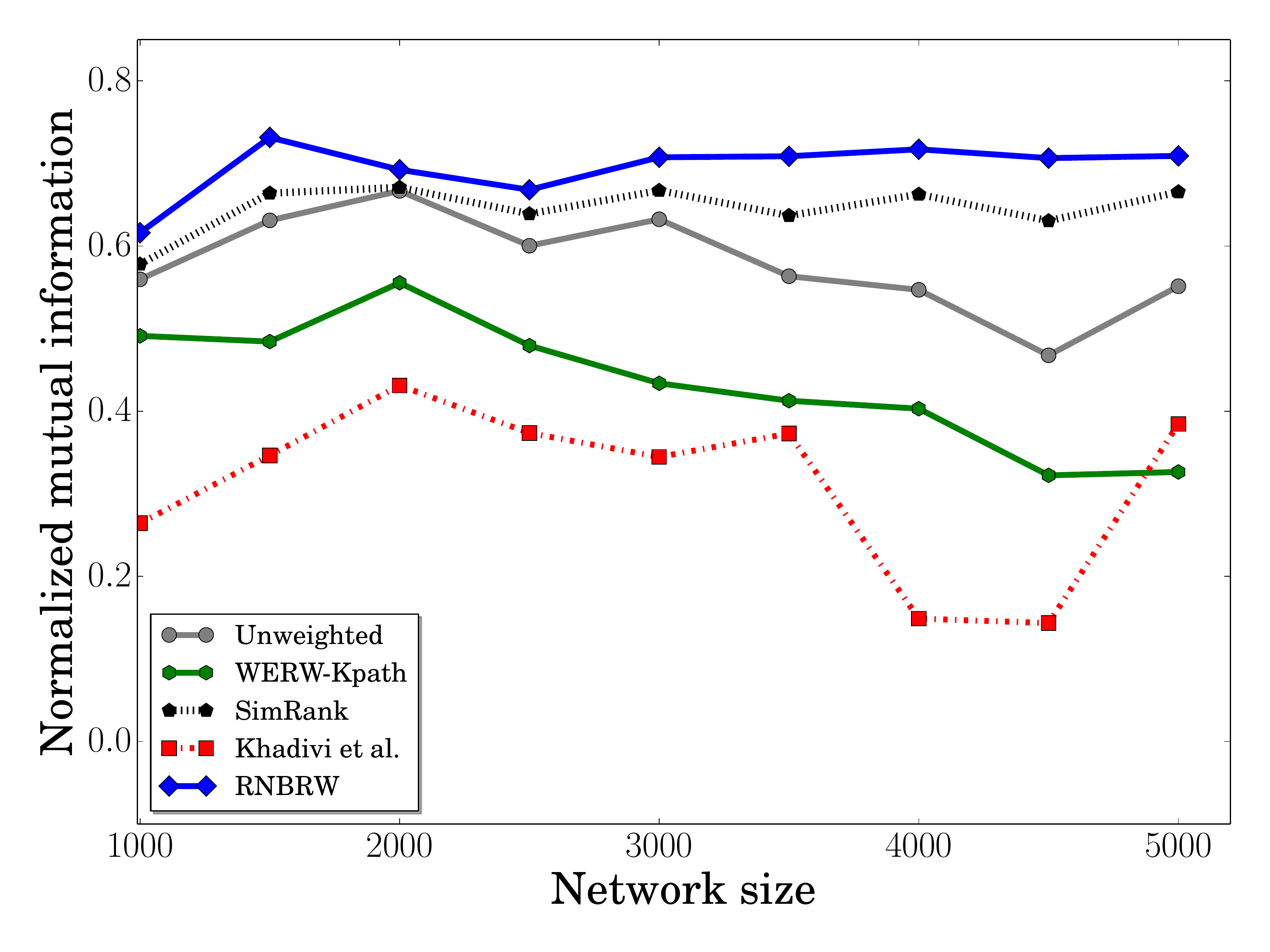}%
		\caption{Performance improvement of Louvain algorithm for different weighting algorithms  for sparse LFR benchmark graphs ($\mu = 0.4$, $\hat{d} = \log(n)$). Loop modulus is omitted in this example due to expensive computational costs.}\label{fig:Khadivi_L}
	\end{figure}


\subsection{Performance of RNBRW to other scalable algorithms}

We now compare the performance of Louvain and CNM equipped with RNBRW to five popular scalable algorithms: Infomap, Label propagation, Edge betweenness, Spinglass, and Walktrap.
We first begin our comparison on a small LFR benchmark graph with identical parameters to that of Figure~\ref{fig:InfoMixModulusNBRW}.
Figure~\ref{fig:InfoMixLFR} suggests that---on an albeit small graph of $n=500$ nodes---Louvain with RNBRW is as good or better at discovering communities as the aforementioned algorithms across all mixing parameters considered.

\begin{figure}
	\begin{center}
		\includegraphics[clip,width=.60\columnwidth]{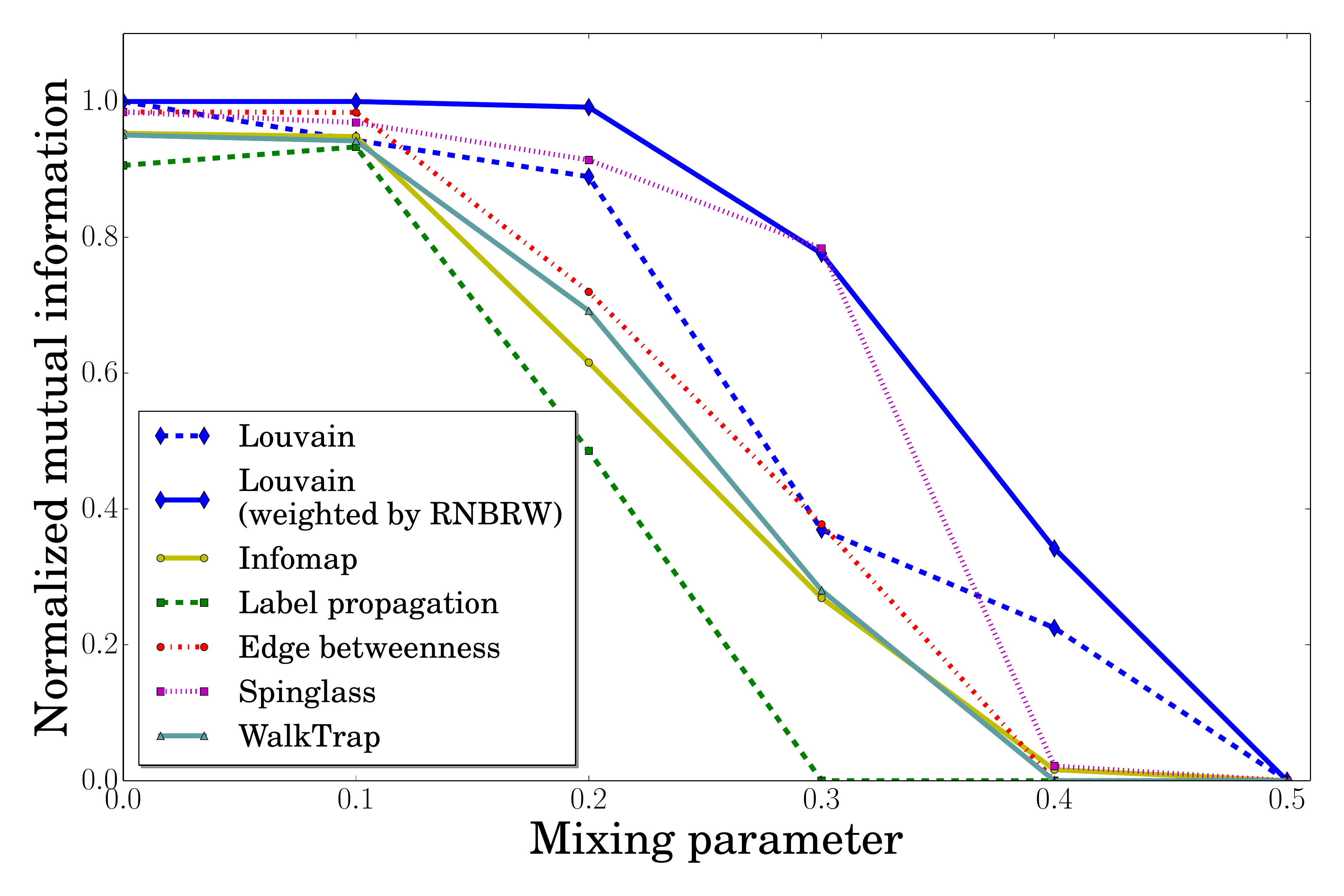}%
	\end{center}
	\caption{Comparing the performance of Louvain equipped with RNBRW to Infomap, Label propagation, Edge betweenness, Spinglass, and Walktrap algorithms. 
	The LFR benchmark networks have  $500$ nodes with average degree $7$, and community sizes ranging from $30$ to $70$.
		}\label{fig:InfoMixLFR}
\end{figure}

We now focus our analysis on larger graphs with varying degrees of sparseness.
Sparse graphs tend to be a particularly challenging case for community detection algorithms.
While \citet{zhao2012consistency} showed consistent detection of communities if the average degree grows at least logarithmically with network size, in practice, some algorithms may still struggle when average degree increases sub-linearly with the number of nodes.
Another problem that arises with modularity maximization methods is their resolution limit in detecting small communities.
This is a critical shortcoming since small size communities are common in large social networks and their size  are  not necessarily growing with graph size. 
\citet{fortunato2007resolution} identify that communities with number of internal edges less than $\sqrt{2|E|}$ are most likely misdetected. 
This leads to a resolution limit that holds back heuristics for  modularity maximization.

We test these community detection algorithms and see how they withstand the challenges of low average degree and existence of small communities for LFR benchmarks.
We consider average degrees that grow logarithmically with the number of nodes.
The community sizes vary between an upper and a lower bound that falls below the resolution limit and the average degree is changing for different size graphs. 

For these larger benchmarks\footnote{We use Beocat for running simulations on the large networks. Beocat is a computer cluster located in Kansas State University \url{https://support.beocat.ksu.edu/BeocatDocs/index.php/Main_Page}.} we use a constant mixing rate $\mu =0.3$ and consider graphs with $10,000$, $100,000$ and $1,000,000$ nodes.
Edge betweenness, Spinglass, and the previously considered weighting algorithms are not included due to prohibitive computational cost with larger network sizes.

We first analyze how the use of RNBRW can improve the performance of Louvain and CNM for sparse graphs.
Figure~\ref{fig:100k} plots the mutual information for graphs of $n = 10,000$ and $n = 100,000$ nodes for CNM and Louvain with and without RNBRW.  
We demonstrate that, even when the average degree is on the order of $\log(n)$, these algorithms equipped with RNBRW achieve a mutual information very close to 1.  
In particular, for $n = 10,000$ and average degree $\log(n)$, equipping CNM with RNBRW increases the mutual information from $0.263$ to $0.974$.


\begin{figure}
	\centering
	\includegraphics[clip,width=1\columnwidth]{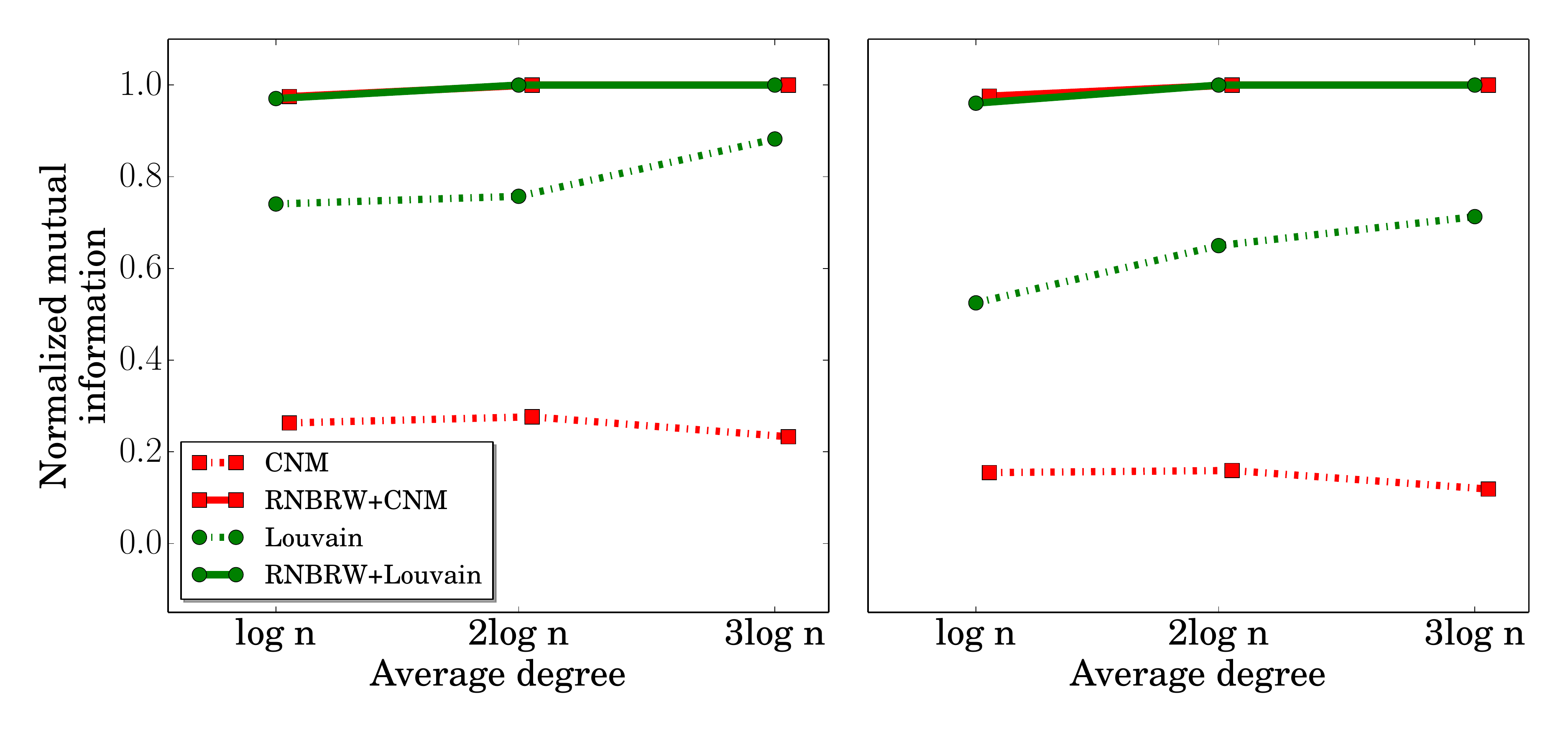}%
	\caption{Plots of NMI for sparse graphs.  
	We use LFR networks with $n = 10,000$ (left) and $n = 100,000$ nodes (right) with mixing parameter $\mu = 0.3$.  
	Average degree varies from $\log(n)$ to $3\log(n)$.
	Weighting by RNBRW leads to substantial performance increases for both the CNM and Louvain algorithms. }\label{fig:100k}
\end{figure}

	\begin{table}
		\resizebox{\textwidth}{!}{
\begin{tabular}{ |l|c|c|c|c|c|c|c| c| }
	\hline
	network size & average degree & Infomap & LP & WT &CNM & Louvain &    RNBRW+CNM&  RNBRW+Louvain\\
	\hline
	\multirow{3}{*}{$10,000$} & $\log n$  & $0.912$ &   $\textbf{0.992}$ & $0.763$ & $0.263$ & $0.740$   &  $0.974$ &   $0.970$\\
	& $2\log n$ &  $\textbf{1}$ &   $0.995$ & $\textbf{1}$  &$0.276$ & $0.757$   &   $\textbf{1}$ &   $\textbf{1}$\\
	& $3\log n$ &   $\textbf{1}$ &   $\textbf{1}$ & $\textbf{1}$ &$0.232$ & $0.882$   &   $\textbf{1}$ &   $\textbf{1}$\\
	\hline	
	\multirow{3}{*}{$100,000$}& $\log n$ & $0.73$ &   $\textbf{0.988}$ & $0.644$ & $0.154$ & $0.524$   &    $0.975$ &   $0.960$\\	
	& $2\log n$ &  $\textbf{1}$ &   $\textbf{1}$ & $\textbf{1}$ & $0.159$ & $0.649$   &   $\textbf{1}$ &   $\textbf{1}$\\	
	& $3\log n$ & $\textbf{1}$ &   $\textbf{1}$ & $\textbf{1}$ &$0.118$ & $0.713$   &     $\textbf{1}$ &   $\textbf{1}$\\	
	\hline			  
	$1,000,000$ & $\log n$ & $0.994$ &   $\textbf{0.998}$ & $-$ & $0.050$ & $0.192$   &   $0.989$ &   $0.969$\\	
	\hline
\end{tabular}}
\caption{Performance measured by NMI for scalable community detection algorithms for sparse LFR networks with $\mu =0.3$ and average degree $\approx \log n$.  
We compare CNM and Louvain with and without RNBRW weighting with Infomap, Label propagation (LP), and Walktrap (WT).  }
\label{tab_results}
\end{table}

We now compare CNM and Louvain equipped with RNBRW with other efficient community detection algorithms.
Summaries of these simulations are be found in Table~\ref{tab_results}. 
We first note that, again, equipping CNM and Louvain with RNBRW improves community detection dramatically; for example, in the 1,000,000 node example, equipping these community detection algorithms with RNBRW leads to an increase NMI in CNM from 0.050 to 0.989 and in Louvain from 0.192 to 0.969. 
Similar increases occur across all other network sizes and average degrees.
Moreover, the detection of communities for CNM and Louvain with RNBRW is comparable to other scalable community detection algorithms.

Although not included in the table, the runtime of RNBRW---especially coupled with Louvain---is much lower than the other considered methods.  
For example, because RNBRW is embarrassingly parallelizable, because Louvain has a small asymptotic runtime, and because RNBRW weights and the Louvain algorithm are performed sequentially, the total runtime of Louvain weighted by RNBRW is smaller than
Infomap $\mathcal{O}(n^2)$, Label propagation $\mathcal{O}(n^2)$, Spinglass $\mathcal{O}(n^{3.2})$ , and 
Walktrap $\mathcal{O}(n^2\log n)$. 
Therefore, adding a RNBRW preprocessing to Louvain seems to be a very reasonable method for obtaining efficient community detection with performance comparable to that of the best scalable community detection algorithms.

\section{Conclusion}

Communities can be identified by the ``richness'' of cycles; more short cycles occur within a community than across communities.  
Hence, existing community detection algorithms may be enhanced by incorporating information about the participation of edges in rhe cyclic structure of the graph. 
We develop renewal non-backtracking random walks (RNBRW) as a way of quantifying the cyclic stucture of a graph.
RNBRW quantifies edge importance as the likelihood of an edge completing a cycle in a non-backtracking random walk, providing a scalable alternative to analyse real-world networks.
We describe how RNBRW weights can be coupled with other popular scalable community detection algorithms, such as CNM and Louvain, to improve their performance.

%

We show that weighting the graph with RNBRW can substantially improve the detection of ground-truth communities.
This improvement is most notable when the network is sparse. 
Furthermore, the RNBRW preprocessing step can overcome the problem of detecting small communities known as resolution limit in modularity maximization methods.
We show that performance appears to be equal or superior to that of other weighting methods, and of other scalable (non-weighted) commnity detection algorithms.
In the case of large, sparse, networks, RNBRW may be quite effective; RNBRW can improve the efficacy of available community detection algorithms without sacrificing their computational efficiency.


\bibliographystyle{chicago}
\bibliography{refs}

\appendix
\section{Results for CNM algorithm}\label{app_CNM}
In section~\ref{sec:results}, we investigate the performance of the Louvain algorithm weighted by RNBRW.  
In this section, we include similar results for the CNM method.

Figure~\ref{fig:InfoMixMlusNBRW} compares the performance of the CNM algorithm equipped with RNBRW to that of other weighing methods (WERW-Kpath, Khadivi, SimRank) as the mixing parameter $\mu$ varies.
Figure~\ref{fig:Khadivi} makes this comparison of performance for sparse graphs as the node size increases.
As before, preprocessing the graph with RNBRW substantially improves the detection of communities when compared to the unweighted algorithm.
Additionally, the improvement provided by RNBRW is substantially larger than that from the WERW-KPath, SimRank, and the Khadavi \textit{et.~al.} methods.

Figure~\ref{fig:InfoMixLFRn} compares CNM equipped with RNBRW to other efficient community detection algorithms (Infomap, Label propagation, Edge betweenness, Spinglass, and Walktrap) across different values of $\mu$.  
Results suggest that CNM with RNBRW is as good or better at discovering communities than the other scalable algorithms across all mixing parameters considered.


\begin{figure}
	\centering		
		\includegraphics[width=.65\textwidth]{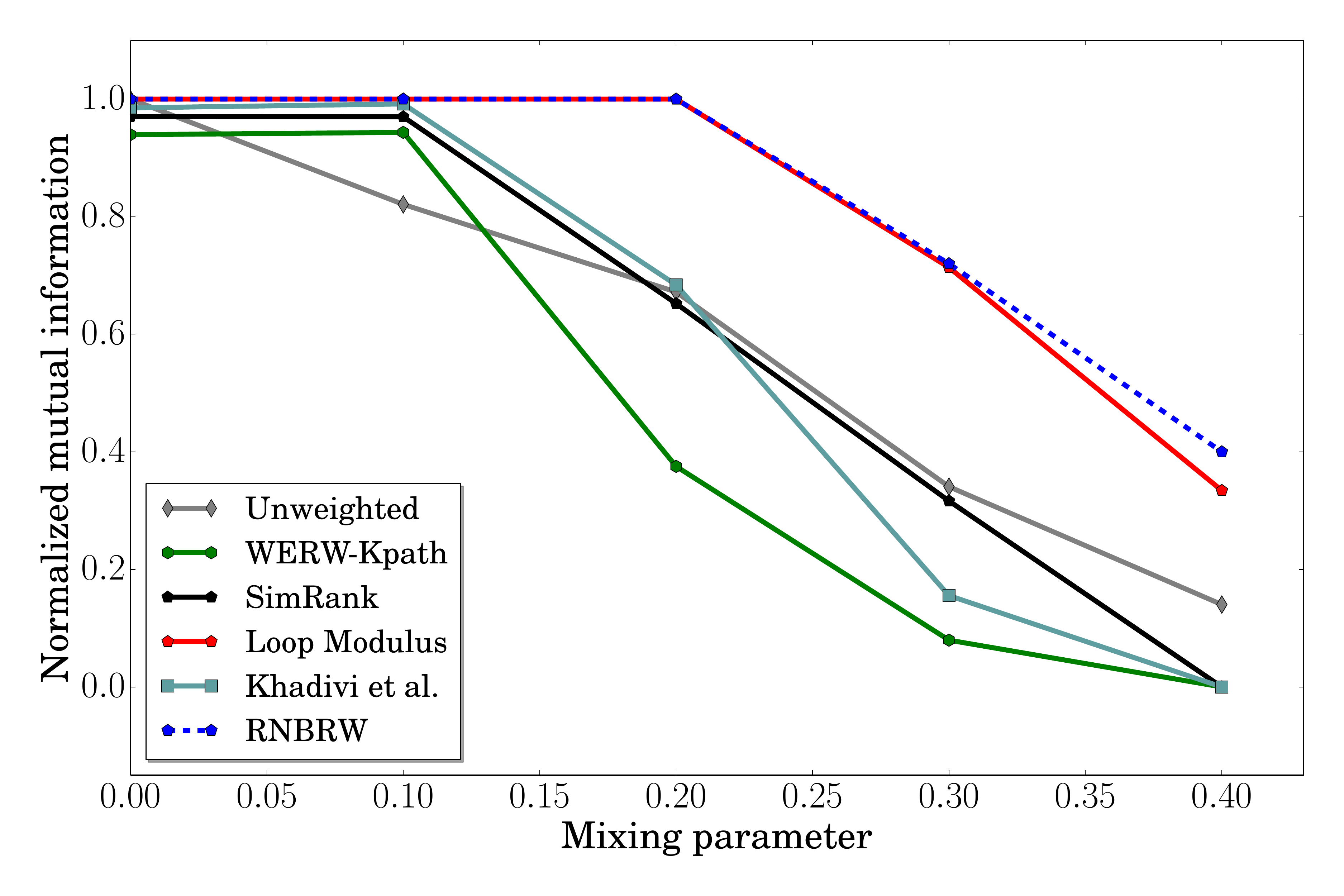}%
	
	\caption{Performance analysis of community detection when coupling weighting methods with the CNM algorithm. 
	We use LFR benchmark networks with  $n = 500$, average degree $7$, and community sizes ranging from $30$ to $70$.
		The mixing rate $\mu$ adjusts the ratio of within-communities links over all links.}\label{fig:InfoMixMlusNBRW}
\end{figure}

\begin{figure}
		\centering
			\includegraphics[clip,width=.6\columnwidth]{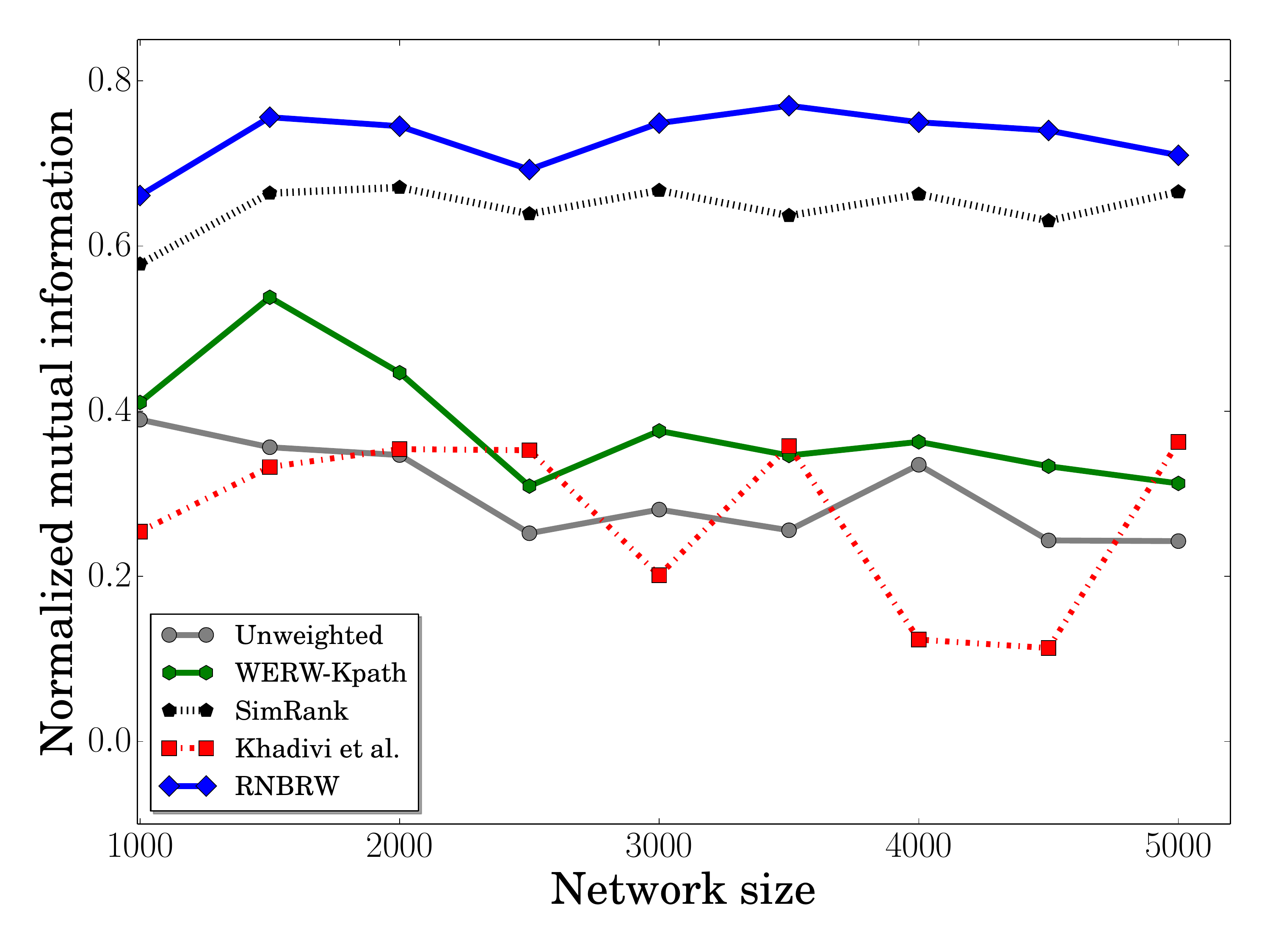}%
		\caption{Performance improvement of CNM algorithm for different weighting algorithms for sparse LFR benchmark graphs ($\mu = 0.4$, $\hat{d} = \log(n)$). Loop modulus is omitted in this example due to expensive computational costs.}\label{fig:Khadivi}
	\end{figure}

\begin{figure}
	\begin{center}
		\includegraphics[clip,width=.60\columnwidth]{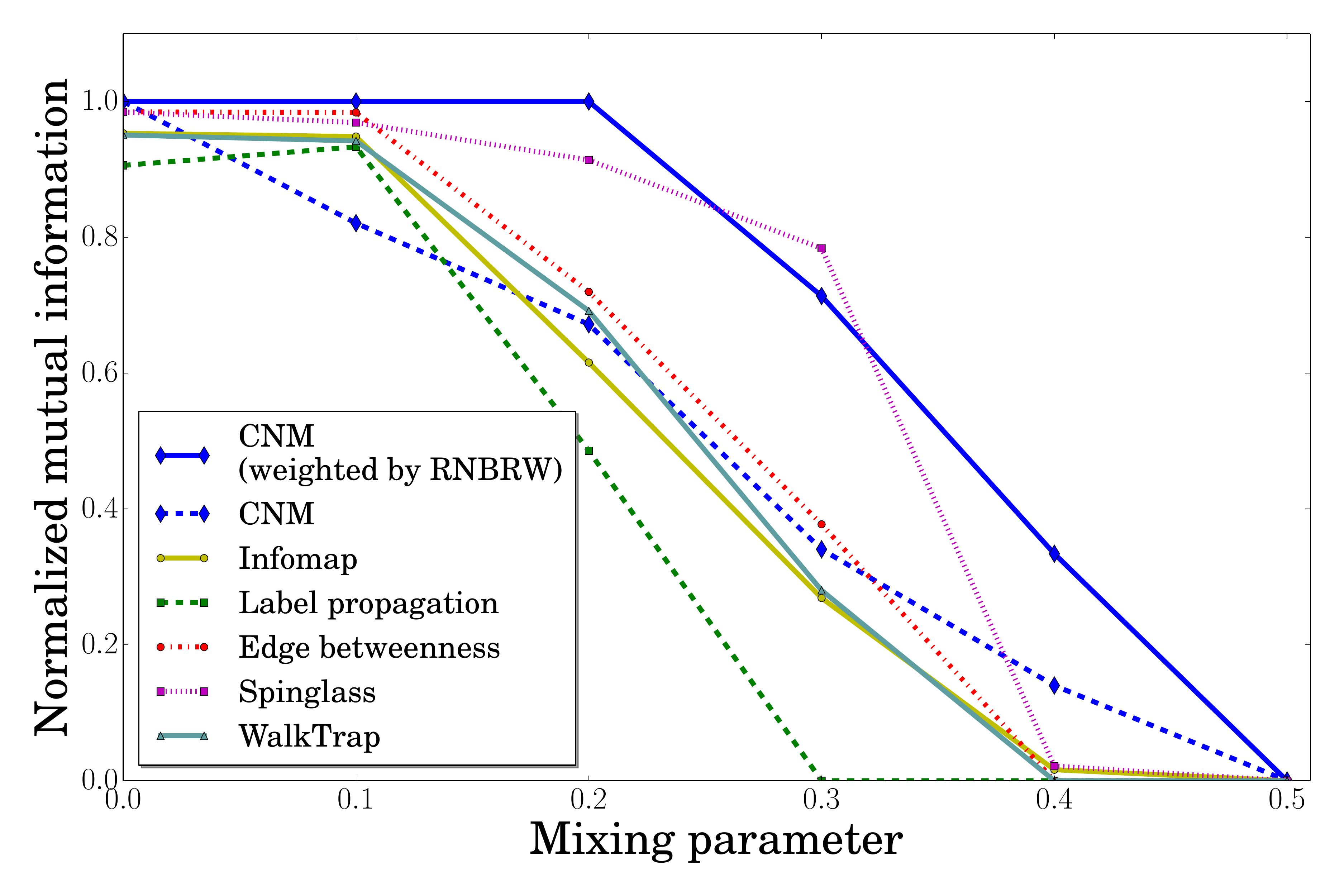}%
	\end{center}
	\caption{Comparing the performance of CNM equipped with RNBRW to Infomap, Label propagation, Edge betweenness, Spinglass, and Walktrap algorithms. 
	The LFR benchmark networks have  $500$ nodes with average degree $7$, and community sizes ranging from $30$ to $70$.
		}\label{fig:InfoMixLFRn}
\end{figure}

	



\end{document}